\documentclass[english,a4paper,12pt,leqno]{article}
\usepackage[T1]{fontenc}
\usepackage[latin9]{inputenc}
\usepackage{babel}
\usepackage{latexsym, amsfonts, amsthm,amssymb,amscd,amsmath,makeidx}
\usepackage{enumerate}
\usepackage{exscale}
\usepackage{color} 
\usepackage{bm}
\usepackage{bbm}
\usepackage{fancyhdr}
\usepackage{mathrsfs} 
\usepackage{enumerate}
\usepackage{natbib}
\usepackage{hyperref}

\newcommand{\ud}{\,\mathrm{d}}

\definecolor{DarkGreen}{rgb}{0.2,0.6,0.2}

\def\green#1{\textcolor{DarkGreen}{#1}}
 \def\green#1{}

\def\om{\omega}

\numberwithin{equation}{section}

\def\ua{\uparrow}

\def\wh{\widehat}
\def\wt{\widetilde}
\def\bbar{\overline}

\def\ignore#1{}

\def\bR{{\mathbb R}}

\def\bC{\mathbb C}

\def\bN{\mathbb N}

\def\bT{\mathbb T}

\def\bP{{\mathbb P}}
\def\bE{{\mathbb E}}

\def\cA{{\mathscr A}}

\def\cS{{\mathscr S}}

\def\cL{{\mathscr L}}

\def\cX{{\mathscr X}}
\def\cY{{\mathscr Y}}

\def\mb#1{\boldsymbol{\mathbf{#1}}}

\def\<{\langle}\def\>{\rangle}

\newtheorem{theorem}{Theorem}[section]
\newtheorem{proposition}[theorem]{Proposition}
\newtheorem{corollary}[theorem]{Corollary}
\newtheorem{lemma}[theorem]{Lemma}

\theoremstyle{definition}
\newtheorem{definition}[theorem]{Definition} 
\newtheorem{example}[theorem]{Example} 
\newtheorem{remark}[theorem]{Remark}

\normalsize

\usepackage{fancyhdr}
\setcitestyle{numbers,open={[},close={]}}

\title{\textbf{Pathwise no-arbitrage in a class\\ of Delta hedging strategies}}
\author{ \normalsize Alexander Schied and Iryna Voloshchenko\footnote{E-mail addresses: {\tt schied@uni-mannheim.de} and {\tt ivoloshc@uni-mannheim.de} \hfill\break
The authors thank Li Chen and Martin U.~Schmidt for helpful comments and gratefully acknowledge support by Deutsche Forschungsgemeinschaft through the Research Training Group RTG 1953.}\\  \normalsize Department of Mathematics\\
 \normalsize University of Mannheim\\
 \normalsize 68131 Mannheim, Germany}

\date{\normalsize First version: January 2, 2016\\
This version: June 14, 2016}

\begin{document}	

\maketitle 

\begin{abstract} 
We consider a strictly pathwise setting for Delta  hedging exotic options, based on F\"ollmer's pathwise It\^o calculus. Price trajectories are $d$-dimensional continuous functions whose pathwise quadratic variations and covariations are determined by a given  local volatility matrix. The existence of Delta hedging strategies in this pathwise setting is established via   existence results for  recursive schemes of parabolic Cauchy problems and via the existence of functional Cauchy problems on path space. Our main results establish the nonexistence of   pathwise arbitrage opportunities in   classes of strategies containing these Delta hedging strategies and under relatively mild conditions on the  local volatility matrix.
 \end{abstract}

\noindent{\bf Keywords:} Pathwise hedging; exotic options; pathwise arbitrage; pathwise It\^o calculus; F\"ollmer integral; local volatility; functional It\^o formula; functional Cauchy problem on path space

\section{Introduction} 

In mainstream finance, the price evolution of a risky asset is usually modeled as a stochastic process defined on some probability space and hence is subject to model uncertainty. 
In a number of situations, however, it is possible to construct continuous-time strategies on a path-by-path basis and without making any probabilistic assumptions on the asset price evolution. A theory of hedging European options of the form $H=h(S(T))$ for one-dimensional asset price trajectories $S=(S(t))_{0\le t\le T}$ was developed by Bick and Willinger~\citep{BickWillinger}
by using F\"ollmer's~\citep{FoellmerIto} approach to pathwise It\^o calculus. Bick and Willinger~\citep{BickWillinger}
 showed, in particular, that if $S$ is strictly positive and admits a pathwise quadratic variation of the form $\<S,S\>(t)=\int_0^ta(s,S(s))\ud s$ for some function $a(s,x)>0$, then a solution $v$ to the terminal-value problem
 \begin{equation}\label{TVP0}
\begin{cases}v\in C^{1,2}([0,T)\times \bR_+)\cap C([0,T]\times \bR_+),\\  \frac{\partial v}{\partial t}+a\frac{\partial^2v}{\partial x^2}=0 \text{ in $[0,T)\times \bR_+$,}\\
v(T,x)=h( x) \text{, $  x\in \bR_+$,}
\end{cases}
\end{equation}
is such that $v(t,S(t))$ is the portfolio value of a self-financing trading strategy that perfectly replicates the option $H=h(S(T))$ in a strictly pathwise sense. In particular, the amount $v(0,S(0))$ can be regarded as the cost required to hedge the option $H$. In continuous-time finance, this amount is usually equated with an arbitrage-free price of $H$. The latter interpretation, however, is not clear in the pathwise situation, because one first needs to exclude the existence of arbitrage in a strictly pathwise sense. 

In the present paper we pick up the approach from~\citep{BickWillinger} and, in a first step, extend their results to a setting with a $d$-dimensional price trajectory, $\mb S(t)=(S_1(t),\dots, S_d(t))^\top$, and an exotic derivative of the form $H=h(\mb S(t_0),\dots, \mb S(t_N))$, where $t_0<t_1<\cdots<t_N$ are the fixing times of daily closing prices and $ h$ is a certain function.
In practice, most  exotic derivatives that pay off at maturity (i.e., European-style) are of this form. Using ideas from~\citep{SchiedStadje}, we show that such options can be hedged in a strictly pathwise sense if a certain recursive scheme of terminal-value problems  \eqref{TVP0} can be solved, and we provide sufficient conditions for the existence and uniqueness of the corresponding solutions. 

In the second part of the paper we then 
approach the absence of strictly pathwise arbitrage within a class of strategies that are based on solutions of recursive schemes of terminal-value problems. This class of strategies hence  includes, in particular,  the Delta hedging strategies of exotic derivatives of the form $H=h(\mb S(t_0),\dots, \mb S(t_N))$. Our main result, Theorem~\ref{no arb thm}, states 
that there are no admissible arbitrage opportunities as soon as the covariation of the price trajectory is of the form 
\begin{equation}\label{qv class strat}
\ud\<S_i,S_j\>=\begin{cases}a_{ij}(t,\mb S(t))\ud t&\text{if $\mb S$ takes values in all of $\bR^d$,}\\
a_{ij}(t,\mb S(t))S_i(t)S_j(t)\ud t&\text{if $\mb S$ takes values in   $\bR^d_+$,}
\end{cases}
\end{equation}
and the matrix $a(t,\mb x)=(a_{ij}(t,\mb x))$ is continuous, bounded, and positive definite. Here,   admissibility refers to the usual requirement that the portfolio value of a strategy must be bounded from below for all considered price trajectories.

Our  result 
on the absence of arbitrage is related to~\citep[Theorem 4]{AlvarezFerrando}, where the absence of pathwise arbitrage is established for $d=1$, $a>0$ constant, and a certain class of smooth strategies. There are, however, several differences between this and our result. First, we consider a more general class of price trajectories that are based on local instead of constant volatility, allow for an arbitrary number $d$ of traded assets, and may either be strictly positive or of the Bachelier type. Second, our class of trading strategies comprises the natural Delta hedging strategies for path-dependent exotic options and, third, we use a completely different approach to prove our result; while Alvarez et al.~\citep{AlvarezFerrando}
use a continuity argument to transfer the absence of arbitrage from the probabilistic Black--Scholes model to a pathwise context, our proof does not rely on any probabilistic asset 
pricing model. Instead, we use Stroock's and Varadhan's idea  for  a probabilistic proof~\citep{StroockVaradhanSupport} of Nirenberg's strong parabolic maximum principle.

We then consider a  setup, in which  an option's payoff may depend on the full trajectory of   asset prices. In this functional framework, F\"ollmer's pathwise It\^o formula needs to be replaced by its functional extension, which was formulated by Dupire~\citep{Dupire} and further developed by Cont and Fourni\'e~\citep{CF}. Furthermore, the Cauchy problem \eqref{TVP0} (and the corresponding  iterated scheme) need to be replaced by a functional Cauchy problem on path space as studied in  Peng and Wang~\citep{Peng2011} and  Ji and Yang~\citep{Ji2013}. We provide versions of our results on hedging strategies and the absence of pathwise arbitrage also in this functional setting.

There are many other approaches to hedging and arbitrage in the face of model risk. For continuous-time results, see, for instance, Lyons~\citep{Lyons95}, Hobson~\citep{Hobson,HobsonNotes}, Vovk~\citep{Vovk11,Vovk12,Vovk14}, Bender et al.~\citep{Benderetal1}, Davis et al.~\citep{DavisRavalObloij}, Biagini et al.~\citep{Biagini}, Beiglb\"ock et al.~\citep{beiglbock} Schied et al.~\citep{SchiedSpeiserVoloshchenko}, and the references therein.  Discrete-time settings were, for instance,  considered in  Acciaio et al.~\citep{Acciaio13}, Bouchard and Nutz~\citep{Bouchard15}, F\"ollmer and Schied~\citep[Section 7.4]{FoellmerSchiedBuch}, Riedel~\citep{Riedel}, and again the references therein.

This paper is organized as follows. In Section~\ref{pathwise hedging  section}, we introduce a general framework for continuous-time trading by means of F\"ollmer's pathwise It\^o calculus~\citep{FoellmerIto}. Based on an extension of an argument from~\citep{FoellmerECM}, our Proposition~\ref{most elementary Ito formula prop} will, in particular, justify the assumption that   price trajectories should admit pathwise quadratic variations and covariations. We will then introduce the  pathwise framework for hedging exotic options \`a la Bick and Willinger~\citep{BickWillinger}. In Section~\ref{arbitrage section}, we will introduce the class of strategies to which our no-arbitrage result, Theorem~\ref{no arb thm}, applies. The extension to the functional setting is given in Section~\ref{Functional Section}. All proofs are contained in  Section~\ref{Proofs section}.

\section{
Strictly pathwise hedging of exotic derivatives}\label{pathwise hedging  section} 

Pathwise It\^o calculus can be used to model financial markets without probabilistic assumptions on the underlying asset price dynamics; see, e.g.,~\citep{Benderetal1,BickWillinger,DavisRavalObloij,FoellmerECM,Lyons95,SchiedCPPI,SchiedStadje,SchiedSpeiserVoloshchenko} for corresponding case studies. 
In this section, we first motivate and describe a general setting for such an approach to asset price modeling and to  the hedging of derivatives. Let us assume that we wish to trade continuously in $d+1$ assets. The first is a riskless bond, $B(t)$, of which we assume for simplicity that it is of the form $B(t)=1$ for all $t$. This assumption can be justified by assuming that we are dealing here only with properly discounted asset prices. The prices of the $d$ risky assets  will be described by  continuous functions $S_1(t),\dots, S_d(t)$, where the time parameter $t$ varies over a certain time interval $[0,T]$. Throughout this paper, we will use vector notation such as $\mb S(t)=(S_1(t),\dots, S_d(t))^\top$. For the moment, when $\mb S=(\mb S(t))_{0\le t\le T}$ is fixed,
a trading strategy will consist of a pair  of  functions $\mb{\xi}=(\xi_1,\dots,\xi_d)^\top$ and $\eta$, where $\xi_i(t)$ describes the number of shares held at time $t$ in the $i^{\text{th}}$ risky asset  and $\eta(t)$ does the same for the riskless asset. The \emph{portfolio value} of $(\mb\xi(t),\eta(t))$ is then given as
\begin{equation}
V(t):=\mb\xi(t) \cdot\mb S(t)+\eta(t)B(t)=\mb\xi(t) \cdot\mb S(t)+\eta(t),\qquad 0\le t\le T,
\end{equation}
where $\mb x\cdot\mb y$ denotes the euclidean inner product of two vectors $\mb x$ and $\mb y$.

A key concept of mathematical finance is the notion of a self-financing trading strategy. 
If trading is only possible at finitely many times $0=t_0<t_1<\cdots<t_N<T$, then $\mb \xi$ and $\eta$ will be constant on each interval $[t_i,t_{i+1})$ and on $[t_N,T]$. In this case it is well-known from discrete-time mathematical finance that  
the trading strategy $(\mb\xi,\eta)$ is self-financing if and only if
\begin{equation}\label{diskreteSFBedingung}
V_{t_i}-V_0=\sum_{k=1}^{i}\mb\xi_{t_{k-1}}(\mb S_{t_k}-\mb S_{t_{k-1}})+\sum_{k=1}^{i}\eta_{t_{k-1}}(B_{t_k}-B_{t_{k-1}}), \qquad
i=1,\dots, N.
\end{equation}
By making the mesh of the partition $\{t_0,\dots,t_N\}$ 
finer and finer, the Riemann sums on the right-hand side of \eqref{diskreteSFBedingung} should converge to corresponding integrals, $\int_0^t\mb\xi(s)\ud \mb S(s)$ and $\int_0^t\eta(s)\ud B(s)$. Clearly, $\int_0^t\eta(s)\ud B(s)$ is a Riemann-Stieltjes integral, and criteria for its existence are well known. For a very specific class of strategies $\mb\xi$, the following proposition gives necessary and sufficient conditions for the existence of  $\int_0^t\mb\xi(s)\ud \mb S(s)$. This proposition extends and elaborates an argument by F\"ollmer~\citep{FoellmerECM}. Before stating this proposition, let us fix for the remainder of this paper a refining sequence of partitions, $(\bT_n)_{n\in\bN}$. That is, each $\bT_n$ is a finite partition of the interval $[0,T]$, and we have $\bT_1\subset\bT_2\subset\cdots$ and the mesh of $\bT_n$ tends to zero as $n\ua\infty$. Moreover, it will be convenient to denote the successor of $t\in\bT_n$ by $t'$. That is, 
$$t'=\begin{cases}\min\{u\in\bT_n\,|\,u>t\}&\text{if $t<T$,}\\
T&\text{if $t=T$.}
\end{cases}
$$

\begin{proposition}\label{most elementary Ito formula prop}Let $t\mapsto \mb S(t)\in\bR^d$ be a continuous function on $[0,T]$. For $i,j\in\{1,\dots, d\}$  and  $K_{ij}\in\bR$ with $K_{ij}=K_{ji}$, we define the trading strategy $\mb\xi^{ij}=(\xi^{ij}_1,\dots,\xi^{ij}_d)^\top$ through
\begin{equation}\label{simple xi strategies}
\xi^{ij}_k(t)=\begin{cases}2\big(S_i(t)+S_j(t)-K_{ij}\big)&\text{if $i\neq j$ and $k=i$ or $k=j$,}\\
2\big(S_i(t)-K_{ii}\big)&\text{if $i= j$ and $k=i$,}\\
0&\text{otherwise.}\end{cases}
\end{equation}
Then  $\int_0^t\mb\xi^{ij}(t)\ud\mb S(t)$ exists for all $t$ and  all $i,j$ as the finite limit of the corresponding Riemann sums, i.e., 
\begin{equation}\label{first def of int dS via Riemann sums}
\int_0^t\mb\xi^{ij}(s)\ud\mb S(s)=\lim_{n\ua\infty}\sum_{s\in\bT_n,\,s\le  t}\mb\xi^{ij}(s)(\mb S(s')-\mb S(s)),
\end{equation}
if and only if the \emph{covariations,}
\begin{equation}\label{quadratic variation eq}
\<S_i,S_j\>(t):=\lim_{n\ua\infty}\sum_{s\in\bT_n,\,s\le  t}(S_i(s')-S_i(s))(S_j(s')-S_j(s)),
\end{equation}
exist in $\bR$ for all $t$ and all $i,j$. In this case it follows that 
\begin{equation}\label{most elementary Ito formula}
\begin{split}
\int_0^t\mb\xi^{ii}(s)\ud\mb S(s)=\big(S_i(t)-K_{ii}\big)^2-\big(S_i(0)-K_{ii})^2-\<S_i,S_i\>(t),
\end{split}
\end{equation}
and, for $i\neq j$,
\begin{equation}\label{most elementary Ito formula 2}
\begin{split}
\int_0^t\mb\xi^{ij}(s)\ud\mb S(s)=\big(S_i(t)+&S_j(t)-K_{ij}\big)^2-\big(S_i(0)+S_j(0)-K_{ij})^2-\sum_{k,\ell\in\{i,j\}}\<S_k,S_\ell\>(t).
\end{split}
\end{equation}
\end{proposition}

The preceding proposition has the following two complementary implications.
\begin{itemize}
\item If one wishes to deal with the very simple strategies of the form \eqref{simple xi strategies}, then one must necessarily assume that the components of the asset price trajectory $\mb S$ admit all pathwise quadratic variations and covariations of the form \eqref{quadratic variation eq}.
\item Suppose that the quadratic variation of $S_i$ exists and vanishes identically. This is, for instance, the case if $S_i$ is Hölder continuous for some exponent $\alpha>1/2$.  Then, for $\mb\xi^{ii}$ as in  \eqref{simple xi strategies} and $K_{ii}=S_i(0)$, the integral $\int_0^t\mb\xi(s)\ud \mb S(s)$ exists for all $t$. By letting $\eta(t):=\int_0^t\mb\xi(s)\ud \mb S(s)-\mb\xi^{\mb S}(t)\cdot\mb S(t)$, we obtain a self-financing trading strategy whose portfolio value is  given by
$V(t)=(S_i(t)-S_i(0))^2$. But this is clearly an arbitrage opportunity as soon as $S_i$ is not constant. Hence, price trajectories of a risky asset necessarily need to be modeled by functions with nonvanishing quadratic variation.
\end{itemize}
These two aspects imply that it is reasonable to require 
that price trajectories $\mb S$ of a risky asset possess all covariations $\<S_i,S_j\>$ in the sense that the limit in \eqref{quadratic variation eq} exists for all $t\in[0,T]$.    It was shown by F\"ollmer~\citep{FoellmerIto} that, if in addition the covariations are continuous functions of $t$, It\^o's formula holds in  a  strictly pathwise sense (see also~\citep{Sondermann} for additional background and an English translation of 
~\citep{FoellmerIto}). Let us thus denote  by $QV^d$ the class of all continuous functions $\mb S:[0,T]\to\bR^d$ on $[0,T]$ for which all covariations $\<S_i,S_j\>(t)$  exist along $(\bT_n)$ and are continuous functions of $t$. We point out that the existence and the value of the covariation $\<S_i,S_j\>(t)$, and hence the space $QV^d$, 
 depend in an essential manner on the choice of the refining sequence of partitions, $(\bT_n)$; see, e.g.,~\citep[p. 47]{Freedman}. Moreover, $QV^d$ is not a vector space~\citep{SchiedTakagi}. It follows easily from  Föllmer's pathwise It\^o formula that for the following class of \lq\lq basic admissible integrands\rq\rq\ $\mb\xi$, the It\^o integral $\int_r^t\mb\xi(s)\ud\mb S(s)$ exists for all $t\in[r,u]\subset[0,T]$ as the finite limit of Riemann sums in \eqref{first def of int dS via Riemann sums}; see~\citep[p.~86]{SchiedCPPI}. This integral is sometimes also called the \emph{F\"ollmer integral}.
 
 \begin{definition}[\bfseries Basic admissible integrands]For $0\le r<u\le T$, an $\bR^d$-valued function $[r,u]\ni t\mapsto\mb\xi(t)$ is called a \emph{basic admissible integrand} for $\mb S\in QV^d$, if there exist $m\in\bN$, a continuous function $\mb A:[r,u]\to\bR^m$ whose components are functions of bounded variation, an open set $O\subset \bR^m\times\bR^d$ such that $(\mb A(t),\mb S(t))\in O$ for all $t$, and a continuously differentiable function $f:O\to\bR$ for which  the function $\mb x\to f(\mb A(t),\mb x)$ is for all $t$ twice continuously differentiable on its domain, such that 
 $$\mb\xi(t)=\nabla_{\mb x}f(\mb A(t),\mb S(t)),
 $$ 
 where  $\nabla_{\mb x}f(\mb a,\mb x)$ denotes the gradient of $\mb x\to f(\mb a,\mb x)$.
 \end{definition}

 Following~\citep{BickWillinger}, we will from now on consider not just one particular price trajectory $\mb S$, but admit an entire class $\cS\subset QV^d$  
 of such trajectories so as to account for the uncertainty of the actual realization of the price trajectory. Specifically, we will consider the classes
 \begin{align*}
 \cS_a:=\Big\{\mb S\in QV^d\,\Big|\,\<S_i,S_j\>(t)=\int_0^ta_{ij}(s,\mb S(s))\ud s\text{ for all $t\in[0,T]$ and $1\le i,j\le d$}\Big\}
 \end{align*}
and
 \begin{align*}
  \cS_a^+:=\Big\{\mb S\in QV^d\,\Big|\,S_i(t)>0,~ \<S_i,S_j\>(t)=\int_0^ta_{ij}(s,\mb S(s))S_i(s)S_j(s)\ud s\text{ for all $t$ and $ i,j$}\Big\},
 \end{align*}
where $a(t,\mb x)=(a_{ij}(t,\mb x))_{i,j=1,\dots,d}$ is a  continuous function of $(t,\mb x)\in[0,T]\times\bR^d$ (respectively of $(t,\mb x)\in[0,T]\times\bR_+^d$ in case of $\cS_a^{+}$) into the set of positive definite symmetric $d\times d$-matrices. Additional assumptions on $a(t,\mb x)$ will be formulated later on. 
 Here, $\bR_+:=(0,\infty)$, and we will write $\bR^d_{(+)}$ to denote the two possibilities, $\bR^d$ and $\bR^d_+$, according to whether we are considering $\cS_a$ or $\cS_a^+$. Similarly, we will write $\cS_a^{(+)}$ etc. Price trajectories in $  \cS_a^+$ can arise as sample paths of multi-dimensional local volatility models. At least for $d=1$, the local volatility function $\sigma(\cdot):=\sqrt{a(\cdot)}$ is  often chosen by calibrating to the market prices of liquid plain vanilla 
options~\citep{DupireSmile}. Since in practice  there are only finitely many given options prices, $\sigma(\cdot)$ is typically only determined on a finite grid~\citep{BuehlerSmile}, and so regularity assumptions on  $\sigma(\cdot)$ can  be made without loss of generality.

Our next goal is to introduce and characterize a class of self-financing trading strategies  that  may depend on the current value of  the  particular realization $\mb S\in\cS_a^{(+)}$ and includes candidates for hedging strategies of European derivatives. Before that, let us introduce some notation. By $C(D)$ we will denote the class of real-valued continuous functions on  a  set $D\subset \bR^n$.
 For an interval $I\subset [0,T]$ with nonempty interior, $\mathring{I}$, we denote by $C^{1,2}(I\times \bR^d_{(+)})$ the class of all functions in $C(I\times \bR^d_{(+)})$ that are continuously differentiable in $(t,\mb x)\in\mathring{I}\times \bR^d_{(+)}$, twice continuously differentiable in $\mb x$ for all $t\in\mathring{I}$, and whose derivatives admit continuous extensions to $I\times \bR^d_{(+)}$.  Let us also introduce the following second-order differential operators,
 $$\cL:=\frac12\sum_{i,j=1}^da_{ij}(t,\mb x)\frac{\partial^2}{\partial x_i\partial x_j}\qquad \text{and}\qquad  \cL^+:=\frac12\sum_{i,j=1}^da_{ij}(t,\mb x)x_ix_j\frac{\partial^2}{\partial x_i\partial x_j}.
 $$

 \begin{proposition}\label{sf prop}Suppose that $0\le r<u\le T$ and that $v\in C^{1,2}([r,u]\times \bR^d_{(+)})$. Then the following conditions are equivalent.
 \begin{enumerate}
 \item For each $\mb S\in\cS_a^{(+)}$, there exists a basic admissible integrand $\bm\xi^{\mb S}$ on $[r,u]$ such that 
 $$v(t,\mb S(t))=v(r,\mb S(r))+\int_r^t\mb\xi^{\mb S}(s)\,\ud \mb S(s)\qquad\text{for $t\in[r,u]$.}
 $$
 \item The function $v$ satisfies the parabolic equation
 \begin{equation}
 \label{heat eq}
 \frac{\partial v}{\partial t}+\cL^{(+)} v=0\qquad\text{in $[r,u]\times \bR^d_{(+)}$.}
 \end{equation}
 \end{enumerate}
 Moreover, if these equivalent conditions hold, then $\bm\xi^{\mb S}$ in {\rm(a)} must necessarily be of the form 
 \begin{equation}\label{xi ident eq}
 \bm\xi^{\mb S}(t)=\nabla_{\mb x}v(t,\mb S(t)).
 \end{equation}
  \end{proposition}

Now suppose that $f:\bR^d_{(+)}\to\bR$ is a continuous function for which  there exists a solution $v$ to the following terminal-value problem,
\begin{equation*}
\begin{cases}v\in C^{1,2}([0,T)\times \bR^d_{(+)})\cap C([0,T]\times \bR^d_{(+)}),\\  \frac{\partial v}{\partial t}+\cL^{(+)} v=0 \text{ in $[0,T)\times \bR^d_{(+)}$,}\\
v(T,\mb x)=f(\mb x) \text{ for $\mb x\in \bR^d_{(+)}$.}
\end{cases}\leqno{(\text{TVP}^{(+)})}
\end{equation*}
For $\mb S\in \cS_a^{(+)}$ and $t\in[0,T)$, we can define  
\begin{equation}\label{hedging strat eq}
\mb\xi^{\mb S}(t):=\nabla_{\mb x}v(t,\mb S(t))\qquad\text{and}\qquad \eta^{\mb S}(t):=v(t,\mb S(t))-\mb\xi^{\mb S}(t)\cdot\mb S(t).
\end{equation}
We then obtain from Proposition~\ref{sf prop} that 
 \begin{equation}
 \begin{split}
\mb\xi^{\mb S}(t)\cdot\mb S(t)+\eta^{\mb S}(t)= v(t,\mb S(t))=v(0,\mb S(0))+\int_0^t\mb\xi^{\mb S}(s)\ud \mb S(s).
 \end{split}
 \end{equation}
Thus,  $(\mb\xi^{\mb S},\eta^{\mb S})$ is a self-financing trading strategy with portfolio value $V^{\mb S}(t)=v(t,\mb S(t))$.  Since the function $v$ is continuous on $[0,T]\times \bR^d_{(+)}$, the limit $V^{\mb S}(T):=\lim_{t\uparrow T}V^{\mb S}(t)$ exists and satisfies 
$$V^{\mb S}(T)=f(\mb S(T))\qquad\text{for all $\mb S\in\cS_a^{(+)}$.}
$$
 In this sense, $(\mb\xi^{\mb S},\eta^{\mb S})$ is a strictly pathwise hedging strategy for the derivative with payoff $f(\mb S(T))$. 
 
 The preceding argument was first made by Bick and Willinger~\citep[Proposition 3]{BickWillinger} in a one-dimensional setting. It is remarkable in several respects. For instance, consider the one-dimensional case with $a(t,x)=\sigma^2x^2$ for some $\sigma>0$ so that $(\text{TVP}^{+})$ becomes the standard Black--Scholes equation, which can be solved for a large class of payoff functions $f$. The preceding argument then shows that the Black--Scholes formula---which is nothing other than an explicit formula for $v(0,S_0)$---can be derived without any probabilistic assumptions whatsoever. It follows, in particular, that 
  the fundamental assumption underlying the Black--Scholes formula is not  the log-normal distribution of asset price returns, but the fact that the quadratic variation of the asset prices is of the form $\<S,S\>(t)=\sigma^2\int_0^tS(s)^2\ud s$. Let us now state  general existence results for solutions of $(\text{TVP})$ and $(\text{TVP}^{+})$, which in the case of $(\text{TVP})$ is taken from Janson and Tysk~\citep{JansonTysk}. Recall that we assume that $a(t,\mb x)$ is positive definite for all $t$ and $\mb x$.
  
  \begin{theorem}\label{JansonTysk thm} Suppose that $f\in C(\bR^d_{(+)})$ has at most polynomial growth in the sense that  $|f(\mb x)|\le c_0(1+|\mb x|^p)$ for some constants $c_0,p>0$. Then, under the following conditions, $(\text{\rm TVP}^{(+)})$  admits a unique solution $v(t,\mb x)$ within the class of functions that are  of at most polynomial growth uniformly in~$t$.
   \begin{enumerate}
  \item  {\bf  (Theorem A.14 in~\citep{JansonTysk})} In case of $(\text{\rm TVP})$, we suppose  that $a_{ij}(t,\mb x)$ is locally Hölder continuous on $[0,T)\times \bR^d$ and that   $|a_{ij}(t,\mb x)|\le c_1(1+|\mb x|^2)$ for a constant $c_1\ge0$, all $(t,\mb x)\in[0,T]\times \bR^d$, and all $i,j$.
  \item In case of $(\text{\rm TVP}^+)$,  we suppose  that $a_{ij}(t,\mb x)$ is bounded and  locally Hölder continuous on $[0,T)\times \bR^d$ for all $i,j$.\end{enumerate}

   \end{theorem}


Our next goal is to extend the preceding hedging argument to the case of a path-dependent \emph{exotic option}. In practice, the payoff of such a derivative is usually of the form 
\begin{equation}\label{exotic eq}
H= h(\mb S({t_0}),\dots, \mb S({t_N}))
\end{equation}
where $0=t_0<t_1<\cdots<t_N= T$ denote the fixing times of daily closing prices and $ h$ is a certain function.

\begin{theorem}\label{exotic Janson Tysk thm}Suppose that the conditions of Theorem~\ref{JansonTysk thm}  are satisfied and $h$ in \eqref{exotic eq} is a locally Lipschitz continuous function on $(\bR_{(+)}^d)^{N+1}$ with a Lipschitz constant that grows at most polynomially. That is,   there exist $p\ge0$ and $L\ge0$ such that, for  $|\mb x_i|,|\mb y_i|\le m$, 
$$\big| h(\mb x_0,\dots,  \mb x_{N})-h(\mb y_0,\dots, \mb y_{N})\big|\le (1+m^p)L\sum_{i=0}^{N}|\mb x_i-\mb y_i|.
$$
 Then, letting
 $$ v_N(t,\mb x_0,\dots, \mb x_N,\mb x):=h(\mb x_0,\dots, \mb x_N)\quad \text{ for} \quad t\in[0,T], \mb x\in \bR^d_{(+)},
 $$ the following recursive scheme for functions $v_k:[t_k,t_{k+1}]\times (\bR^d_{(+)})^{k+1}\times \bR^d_{(+)}\to\bR$, for $ k=0,\dots, N-1,$ is well-defined.
\begin{itemize}
\item For $k=N-1, N-2,\dots, 0$, the function  $f_{k+1}(\bm x):= v_{k+1}(t_{k+1},\mb x_0,\dots, \mb x_k,\mb x,\mb x)$ is continuous in $\mb x$, and $(t,\mb x)\mapsto v_k(t,\mb x_0,\dots, \mb x_k,\mb x)$ is the solution of $(\text{TVP}^{(+)})$ with terminal condition  $f_{k+1}$ at   time $t_{k+1}$.
\end{itemize}
\end{theorem}

The  condition on the local Lipschitz continuity  of $h$ in the preceding result can often be relaxed in more specific situations. Examples are the pathwise versions of the ($d$-dimensional) Bachelier and Black--Scholes models, which both correspond to the choice  $a_{ij}(t,\mb x)=\wt a_{ij}$ for a constant positive definite matrix $(\wt a_{ij})$. In these cases the recursive scheme in Theorem~\ref{exotic Janson Tysk thm} can be solved for large classes of payoff functions $h$ without requiring local Lipschitz continuity. As a matter of fact, even the continuity of $h$ can be relaxed so as to account for discontinuous payoffs as, e.g., in barrier options. This also applies to the strictly pathwise hedging argument that we are going to formulate next. However, these relaxations need case-by-case arguments. We therefore do not spell them out explicitly here and leave the details to the interested reader.

Now let $H$ be an exotic option as in \eqref{exotic eq} and suppose that the recursive scheme in Theorem~\ref{exotic Janson Tysk thm} holds for functions $v_k$, $k=0,\dots, N$. When denoting by $\nabla_{\mb x}v_k$ the gradient of the function $\mb x\mapsto v_k(t,\mb x_0,\dots,\mb x_k,\mb x)$, then
 \begin{equation}\label{exotic hedging strat eq}
 \begin{split}
\mb\xi^{\mb S}(t)&:=\nabla_{\mb x}v_k(t,\mb S(t_0),\dots,\mb S(t_k),\mb S(t)),\\
 \eta^{\mb S}(t)&:=v_k(t,\mb S(t_0),\dots,\mb S(t_k),\mb S(t))-\mb\xi^{\mb S}(t)\cdot\mb S(t),
\end{split}\qquad\qquad \text{for } t\in[t_k,t_{k+1}),
\end{equation}
is a self-financing trading strategy on each interval $[t_k,t_{k+1})$ in the sense that 
$$\mb\xi^{\mb S}(t)\cdot\mb S(t)+\eta^{\mb S}(t)= v_k(t,\mb S(t_0),\dots,\mb S(t_k),\mb S(t))=v_k(t_k,\mb S(t_0),\dots,\mb S(t_k),\mb S(t_k))+\int_{t_k}^t\mb\xi^{\mb S}(s)\ud \mb S(s).
$$
The continuity of $t\mapsto v_k(t,\mb S(t_0),\dots,\mb S(t_k),\mb S(t))$ implies the existence of the limit
$$\int_{t_k}^{t_{k+1}}\mb\xi^{\mb S}(s)\ud \mb S(s):=\lim_{t\uparrow t_{k+1}}\int_{t_k}^t\mb\xi^{\mb S}(s)\ud \mb S(s)
$$
and hence allows us to define 
\begin{equation}\label{Ito int extension eq}
\int_{0}^t\mb\xi^{\mb S}(s)\ud \mb S(s):=\sum_{k=0}^{\ell-1}\int_{t_k}^{t_{k+1}}\mb\xi^{\mb S}(s)\ud \mb S(s) +\int_{t_\ell}^t\mb\xi^{\mb S}(s)\ud \mb S(s), \qquad t\in[0,T],
\end{equation}
where $\ell$ is the largest $k$ such that $t_k<t$.
 With these conventions, we obtain the following Delta hedging result.

\begin{corollary}\label{exotic hedge cor}Let $H$ be an exotic option as in \eqref{exotic eq} and suppose that the recursive scheme in Theorem~\ref{exotic Janson Tysk thm} holds for functions $v_k$, $k=0,\dots, N$. Then, for each $\mb S\in\cS_a^{(+)}$,  the strategy \eqref{exotic hedging strat eq}
 is self-financing in the above sense and satisfies
 $$\lim_{t\uparrow T}\mb\xi^{\mb S}(t)\cdot\mb S(t)+\eta^{\mb S}(t)=v_0(0,\mb S(t_0))+\int_0^T\mb\xi^{\mb S}(s)\ud \mb S(s)=h(\mb S({t_0}),\dots, \mb S({t_N})).
 $$
 In this sense, $(\mb \xi^{\mb S},\eta^{\mb S})$ is a strictly pathwise Delta hedging strategy for $H$.
\end{corollary}

The preceding corollary establishes a general, strictly  pathwise hedging result for a large class of exotic options arising in practice. It also identifies $v_0(0,\mb S(t_0))$ as the amount of cash needed at $t=0$ so as to perfectly replicate the payoff $H$ for all price trajectories in $\cS_a^{(+)}$. In continuous-time finance, this amount is usually equated with an \emph{arbitrage-free 
price} for $H$. In our situation, however, the interpretation of $v_0(0,\mb S(t_0))$ as an arbitrage-free price lacks   an essential ingredient: We do not know whether our class of trading strategies is indeed arbitrage-free with respect to all possible price trajectories in $\cS_a^{(+)}$. This question will now be explored in the subsequent section. Our corresponding result, Theorem~\ref{no arb thm}, gives sufficient conditions under which trading strategies, as those in Corollary~\ref{exotic hedge cor}, do indeed not generate arbitrage in our pathwise framework. 
Theorem~\ref{no arb thm} will be the main result of this paper.

\begin{remark}[\bfseries Robustness of the hedging strategy]\label{robust remark}The strategy \eqref{exotic hedging strat eq}
 yields a perfect hedge for the exotic option $H$ only if the actually realized price trajectory, $ \mb S$, belongs to the set $\cS_a^{(+)}$. In reality, however, the realized quadratic variation is typically subject to uncertainty, and therefore it may turn out \emph{a posteriori} that $\mb S$ does actually not belong to $\cS_a^{(+)}$. If $\mb S$ nevertheless belongs to $QV^d$, one can then speak of volatility uncertainty. One possible approach to  volatility uncertainty was developed in~\citep{Lyons95}, where, for the case in which $H=h(\mb S(T))$,  the linear equation $(\text{TVP}^{+})$ is replaced by a certain nonlinear partial differential equation that corresponds to a worst-case approach within a class of price trajectories whose realized volatility may vary within a given set. A different approach to volatility uncertainty was proposed in~\citep{RobustBS} for the case $d=1$, in which we write $S$ instead of the vector notation $\mb S$.   Although~\citep{RobustBS} is set up in a diffusion framework, it is straightforward to translate the comparison result of~\citep[Theorem 6.2]{RobustBS} into a strictly pathwise framework. For options of the form $H=h(S(T))$ with $h\ge0$ convex, one then gets that the Delta hedge \eqref{exotic hedging strat eq} is  \emph{robust} in the sense that it is still a superhedge as long as $a$ overestimates the realized quadratic variation, i.e.,  $\int_r^ta(s,S(s))\ud s\ge \<S,S\>(t)-\<S,S\>(r)$ for $0\le r\le t\le T$. Thus, if a Delta hedging strategy is robust, then a trader
can monitor its performance by comparing $a(t,S(t))$ to the
realized quadratic variation $\<S,S\>$. In~\citep{SchiedStadje}, it was analyzed to what extent the preceding result can be extended to exotic payoffs of the form $H= h(  S({t_0}),\dots,  S({t_N}))$. It was shown that robustness then breaks down  for a large class of relevant convex payoff functions $h$, but that it still holds if $h$ is directionally convex. 
\end{remark}

\section{Absence of pathwise arbitrage}\label{arbitrage section}

We are now going to study the absence of pathwise arbitrage within a class of strategies that is suggested by the pathwise Delta hedging strategies constructed in Theorem~\ref{exotic Janson Tysk thm} and Corollary~\ref{exotic hedge cor}. We refer to the   paragraph   preceding Remark~\ref{robust remark} for a motivation of this problem. Let us first introduce the class of strategies we will consider.

\begin{definition}\label{strategies def}
Suppose that $N\in\bN$, $0=t_0<t_1<\cdots <t_N=t_{N+1}=T$, and $v_k$ $(k=0,\dots,N)$ are real-valued  continuous functions  on $[t_k,t_{k+1}]\times (\bR^d_{(+)})^{k+1}\times \bR^d_{(+)}$ such that, for $k=0,\dots, N-1$, the function    $(t,\mb x)\mapsto v_k(t,\mb x_0,\dots, \mb x_k,\mb x)$ is the solution of $(\text{TVP}^{(+)})$ with terminal condition  $f_{k+1}(\bm x):= v_{k+1}(t_{k+1},\mb x_0,\dots, \mb x_k,\mb x,\mb x)$ at   time $t_{k+1}$. For $\mb S\in\cS_a^{(+)}$, we then define $\mb\xi^{\mb S}$ as in \eqref{exotic hedging strat eq} and 
\begin{equation}
V^{\mb S}_{\mb\xi}(t):=v_0(0,\mb S(0))+\int_0^t\mb\xi^{\mb S}(s)\ud \mb S(s),
\end{equation}
where the pathwise It\^o integral is understood as in \eqref{Ito int extension eq}. By  $\cX^{(+)}$ we denote the collection of all pairs $(v_0(0,\cdot),\mb \xi^\cdot)$ that arise in this way.
\end{definition}

Theorem~\ref{exotic Janson Tysk thm} gives   sufficient conditions for the existence of a family of functions $(v_k)$ as in the preceding definition, but these conditions are not necessary. In particular, as mentioned above, the local Lipschitz continuity of the terminal function $v_N$ can be relaxed in many situations. 
We can now define our strictly pathwise  notion of an arbitrage strategy.

\begin{definition}\label{arbitrage def}  \textbf{((Admissible) arbitrage opportunity)} A pair $(v_0(0,\cdot),\bm\xi^\cdot)\in\cX^{(+)}$ is called an \emph{arbitrage opportunity} for $\cS_a^{(+)}$ if the following conditions hold.
\begin{enumerate}
\item $V^{\mb S}_{\mb\xi}(T)\ge0$ for all $\mb S\in\cS_a^{(+)}$.
\item There exists at least one $\mb S\in\cS_a^{(+)}$  for which $V^{\mb S}_{\mb\xi}(0)=v_0(0,\mb S(0))\le 0$ and $V^{\mb S}_{\mb\xi}(T_0)>0$ for some $T_0\in(0,T]$.
\end{enumerate}
An arbitrage opportunity $(v_0(0,\cdot),\bm\xi^\cdot)$ will be called \emph{admissible} if the following condition is also satisfied.
\begin{enumerate}
\setcounter{enumi}{2}
\item There exists a constant $c\ge 0$ such that $V^{\mb S}_{\mb\xi}(t)\ge-c$ for all $\mb S\in\cS_a^{(+)}$ and $t\in[0,T]$.
\end{enumerate}
\end{definition}

Let us comment on the preceding definition. Condition  (a) states that one can follow the strategy $(v_0(0,\cdot),\bm\xi^\cdot)$ up to time $T$ without running the risk of ending up with negative wealth at the terminal time. Now let $\mb S$ be as in condition (b). The initial spot value, $\mb S_0:=\mb S(0)$,  will then be such that $v_0(0,\mb S_0)=V^{\mb S}_{\mb\xi}(0)\le 0$. Hence, for any price trajectory $\wt{\mb S}\in\cS_a^{(+)}$ with $\wt{\mb S}(0)=\mb S_0$, only a  nonpositive initial investment $v_0(0,\mb S_0)=V^{\wt{\mb S},\mb\xi}(0)$ is required so as to end up with the nonnegative terminal wealth $V^{\wt{\mb S},\mb\xi}(T)\ge0$. Moreover, for the particular price trajectory $\mb S$, there exists a time $T_0$ at which one can make the strictly positive profit $V^{\mb S}_{\mb\xi}(T_0)>0$. This profit can be locked in, e.g.,  by halting all  trading from time $T_0$ onward. In this sense, the strategy  $(v_0(0,\cdot),\bm\xi^\cdot)$ is indeed an arbitrage opportunity. Condition (c) is a constraint on the strategy  $(v_0(0,\cdot),\bm\xi^\cdot)$ that is analogous to the   admissibility  constraint that is usually imposed in continuous-time probabilistic models so as to exclude  doubling-type strategies. Indeed, it follows, e.g., from Dudley's result~\citep{Dudley} that standard diffusion models typically admit arbitrage opportunities in the class of strategies whose value process is not bounded from below (see also the discussion in~\citep[Section 1.6.3]{JeanblancYorChesney}). In our pathwise setting, an example of an  arbitrage opportunity that does not satisfy condition (c) will be provided in Example~\ref{arb ex} below. First, however, let us state the main result of our paper.

\begin{theorem}[\bfseries{Absence of admissible arbitrage}]\label{no arb thm}
Suppose that $a(t,\mb x)$ is continuous, bounded, and positive definite for all $(t,\mb x)\in[0,\wt T]\times\bR_{(+)}^d$, where $\wt T>T$. Then there are no admissible arbitrage opportunities in $\cX^{(+)}$.
\end{theorem}

\begin{example}\label{arb ex} \textbf{(A non-admissible arbitrage opportunity)} Suppose that $d=1$ and $a\equiv 2$. Then the assumptions of Theorem~\ref{no arb thm} are clearly satisfied. Moreover, $\cL=\partial^2/\partial x^2$ and $(\text{TVP})$
is the  time-reversed Cauchy problem for the standard heat equation. There are many explicit examples of nonvanishing functions $v$ satisfying   $(\text{TVP})$ with terminal condition $f\equiv0$; see, e.g.,~\citep[Section II.6]{WidderHeatEquation}. By Widder's uniqueness theorem for nonnegative solutions of the heat equation,~\citep[Theorem VIII.2.2]{WidderHeatEquation},  any such function $v$ must be unbounded from above and from below on every nontrivial strip  $[t,T]\times\bR$ with $t<T$. In particular, there must be $0\le t_0<t_1<T$  and $x_0,x_1\in\bR$ such that $v(t_0,x_0)=0$ and $v(t_1,x_1)>0$. By means of a time shift, we can assume without loss of generality that $t_0=0$.  It can be shown easily that $\cS_a$ contains trajectories that can connect the two points $x_0$ and $x_1$ within time $t_1-t_0$, and so it follows that the function $v$ gives rise to an arbitrage opportunity.\end{example}

 \section{Extension to functionally dependent strategies }\label{Functional Section} 
 Recall from \eqref{exotic eq}  our representation $H= h(\mb S({t_0}),\dots, \mb S({t_N}))$ of the payoff of an exotic option, based on asset prices sampled at the $N+1$ dates $0=t_0<t_1<\cdots<t_N= T$. If $N$ is large, it may be convenient to use a continuous-time approximation of the payoff $H$. For instance, the payoff $H=(\frac1{N}\sum_{n=1}^NS^1_{t_n}-K)^+$ of an average-price Asian call option on the first asset, $S^1$, can be approximated by a call option based on a continuous-time average of asset prices, $H\approx (\frac1T\int_0^TS^1_t\,dt-K)^+$. Approximations of this type may be easier to treat analytically and are standard in the textbook literature.  In this section, we extend our preceding results to a situation that covers such continuous-time approximations of \eqref{exotic eq}.  That is, we will consider payoffs of the form $H(\mb S),$ where $\mb S$ describes the entire path of the underlying price trajectory up to time $T$, and $H$ is a suitable mapping from the Skorohod space $D([0,T],\bR^d)$ to $\bR.$    
   This will involve functional It\^o calculus as introduced  by~\cite{Dupire} and further developed by Cont and Fourni\'e~\cite{CF}. In the sequel, we will use the same notation as in~\cite{CF}. 
   
   For a $d$-dimensional c\`{a}dl\`{a}g path $\mb X$ in the Skorohod space  $D([0,T], \mathbb{R}^d)$ we write $\mb X(t)$  for the value of $\mb X$ at time $t$ and $\mb X_t=(\mb X(u))_{0\leq u\leq t}$ for the restriction of $\mb X$ to the interval $[0,t]$. Hence, $\mb X_t\in D([0,t], \mathbb{R}^d)$. We will work with non-anticipative functionals as defined in~\citep[Definition 1]{CF}, i.e., with a family $F=\left(F_t\right)_{t\in[0,T]}$ of maps 
$
F_t:D([0,t], \mathbb{R}^d)\mapsto\mathbb{R}.
$
For all further notation and relevant definitions, we refer to~\citep[Section 1]{CF}.
 
The functional It\^o formula (in the form of~\cite[Theorem 3]{CF}) yields that we can define  general admissible integrands $\mb\xi$ in the following way, so as to ensure that the pathwise  It\^o integral $\int_r^t\mb\xi(s)\ud\mb S(s)$ exists for all $t\in[r,u]\subset[0,T]$ as a finite limit of Riemann sums; see~\cite[p.~1051]{CF}.

 \begin{definition}[\bfseries General admissible integrands]
Suppose that  $0\le r<u\le T$,  $m\in\bN$,  $\mb V:[r,u]\to\bR^m$ is c\`adl\`ag and satisfies $\sup_{t\in[r,u]\textbackslash \mathbb{T}_n\cap [r,u]}\arrowvert \mb V(t)-\mb V(t-) \arrowvert\to 0$, and  $F$ is a non-anticipative functional in $ \mathbb{C}^{1,2}([r,u])$  (see~\citep[Definition 9 ]{CF})  
  such that the following regularity conditions are satisfied:
 \begin{itemize}
 \item[(a)]  $F$ depends in a predictable manner on its  second argument $\mb V$, i.e.,
 $$
 F_t(\mb X_t, \mb V_t)= F_t(\mb X_t, \mb V_{t-}),
 $$
 where  $\mb V_{t-}$ denotes the path defined on $[r,t]$ by
 $$
 \mb V_{t-}(s)=\mb V(s),\quad s\in[r,t),\quad \mb V_{t-}(t)=\mb V(t-),
 $$
 \item[(b)] $F,$ its vertical derivative $\nabla_{\mb x}F,$ and its second vertical derivative $\nabla^2_{\mb x}F$ belong to the class $\mathbb{F}^\infty_l$ (see~\citep[Definition  3]{CF}),  \item[(c)]   the horizontal derivative $\mathcal{D} F$ as well as the second vertical derivative $\nabla^2_{\mb x}F$ of $F$ satisfy the local boundedness condition~\citep[equation (9)]{CF}. 
 \end{itemize}
 Then  $$ \mb \xi(t)=\nabla_{\mb x}F_t(\mb S_{[r,u],t}, \mb V_t)$$
is called a  \emph{general admissible integrand} for $\mb S\in QV^d$. Here, $\mb S_{[r,u]}$ denotes the restriction of $\mb S$ to the interval $[r,u].$
  \label{Def general adm integrands} \end{definition}

In analogy to Proposition~\ref{sf prop}, we will now characterize self-financing trading strategies  that  may depend on the entire past evolution of  the  particular realization $S\in\cS_a^{(+)}$.  For an interval $I\subset [0,T]$ with nonempty interior, $\mathring{I}$, we denote by $\bC^{1,2}(I)$ the class of all non-anticipative functionals on $\bigcup_{t\in[a,b]} 
 D([a,t], 
 \bR^d_{(+)})$  that are horizontally differentiable and twice vertically differentiable on  $\mathring{I}$ and whose derivatives are continuous at fixed times and admit continuous extensions to $I$. 
Thus, lifting the second-order differential operators $\cL$ and $\cL^+$ yields the following operators on path space,
 $$\cA:=\frac12\sum_{i,j=1}^da_{ij}(t,\mb X(t))\nabla^2_{ij}\qquad \text{and}\qquad  \cA^+:=\frac12\sum_{i,j=1}^da_{ij}(t,\mb X(t))X_i(t)X_j(t)\nabla^2_{ij}.  $$
 The following proposition is a functional version of Proposition~\ref{sf prop}. 
  
  \begin{proposition}\label{sf prop functional}Suppose that $0\le r<u\le T$ and let  $F\in \mathbb{C}^{1,2}([r,u])$ 
   be a non-anticipative functional satisfying the conditions from Definition~\ref{Def general adm integrands}. 
 Then the following conditions are equivalent.
 \begin{enumerate}
 \item For each $\mb S\in\cS_a^{(+)}$, there exists a general admissible integrand $\bm\xi^{\mb S}$ on $[r,u]$ such that 
 $$F_t(\mb S_{[r,u],t})=F_r(\mb S_{[r,u],r})+\int_r^t\mb\xi^{\mb S}(s)\,\ud \mb S(s)\qquad\text{for $t\in[r,u]$.} 
 $$
 \item The functional $F$ satisfies the 
 path-dependent parabolic equation
 \begin{equation}
 \label{heat eq functional}
\mathcal{D} F+\cA^{(+)} F= 0\qquad\text{on $  \cS_a^{(+)}\big \arrowvert_{[r,u]}$.} 
\end{equation}  
 \end{enumerate}
 Moreover, if these equivalent conditions hold, then $\bm\xi^{\mb S}$ in {\rm(a)} must necessarily be of the form 
 \begin{equation}\label{xi ident eq functional}
 \bm\xi^{\mb S}(t)=\nabla_{\mb x}F_t(\mb S_{[r,u],t}).
 \end{equation}
  \end{proposition}

Now suppose that for  suitably given  $H:D([0,T],\mathbb{R}^d_{(+)})\to \mathbb{R}$ 
there exists a solution $F$ to the following 
path-dependent terminal-value problem,
\begin{equation*}
\begin{cases}F\in\mathbb{C}^{1,2}([0,T))&\text{satisfies the conditions  from Definition~\ref{Def general adm integrands}},\\ \mathcal{D} F+\cA^{(+)} F=0&\text{in $\bigcup_{t\in[0,T)}D([0,t], \bR^d_{(+)}),$}\\
F_T(\mb X_T)=H(\mb X_T)& \text{for $\mb X_T\in D([0,T],\bR^d_{(+)})$.}
\end{cases}\leqno{( \text{FTVP}^{(+)})}
\end{equation*}
Note that the terminal condition $H$ has to be defined on the Skorohod space $D([0,T],\bR^d_{(+)})$ as opposed to $C([0,T],\bR^d_{(+)})$, because we need to take its vertical derivatives, which requires applying discontinuous shocks.

Then, for $\mb S\in \cS_a^{(+)}$ and $t\in[0,T)$, we can define  
\begin{equation}\label{hedging strat eq functional}
\mb\xi^{\mb S}(t):=\nabla_{\mb x}F_t(\mb S_t)\qquad\text{and}\qquad \eta^{\mb S}(t):=F_t(\mb S_t)-\mb\xi^{\mb S}(t)\cdot\mb S(t).
\end{equation}
Proposition~\ref{sf prop functional}  gives 
 \begin{equation}
 \begin{split}
\mb\xi^{\mb S}(t)\cdot\mb S(t)+\eta^{\mb S}(t)= F_t(\mb S_t)=F_0(\mb S_0)+\int_0^t\mb\xi^{\mb S}(s)\ud \mb S(s),
 \end{split}
 \end{equation}
whence we infer that  $(\mb\xi^{\mb S},\eta^{\mb S})$ is a self-financing trading strategy with portfolio value $V^{\mb S}(t)=F_t(\mb S_t)$.  Since the functional $F$ is left-continuous on $[0,T]$ and $\mb S$ is continuous, the limit $V^{\mb S}(T):=\lim_{t\uparrow T}V^{\mb S}(t)$ exists and satisfies 
$$V^{\mb S}(T)=H(\mb S)\qquad\text{for all $\mb S\in\cS_a^{(+)}$.}
$$
 Thus, $(\mb\xi^{\mb S},\eta^{\mb S})$ is a strictly pathwise hedging strategy for the derivative with payoff $H=H(\mb S).$ 
\ignore{
\begin{corollary}\label{hedge cor functional}Let $H=H(\mb S_T)$ be an exotic path-dependent option with a suitably given $H:C([0,T],\mathbb{R}^d_{(+)})\to \mathbb{R}$ and suppose that  there exists a solution $F$ to the terminal-value problem $( \text{FTVP}^{(+)})$. Then, for each $\mb S\in\cS_a^{(+)}$,  the strategy \eqref{hedging strat eq functional}
 is self-financing  and satisfies
 $$\lim_{t\uparrow T}\mb\xi^{\mb S}(t)\cdot\mb S(t)+\eta^{\mb S}(t)=F_0(\mb S_0)+\int_0^t\mb\xi^{\mb S}(s)\ud \mb S(s)=H(\mb S_T).
 $$
Thus, $(\mb \xi^{\mb S},\eta^{\mb S})$ is a strictly pathwise Delta hedging strategy for $H$.
\end{corollary}}

In the next step, we will explore conditions yielding the existence and uniqueness of solutions to $(\text{FTVP})$ and $(\text{FTVP}^{+})$. 
 Path-dependent PDEs such as \eqref{heat eq functional} are closely related to 
backward stochastic differential equations (BSDEs) generalizing the (functional) Feynman-Kac formula~\citep{Dupire}. In~\citep{Peng2011},  a one-to-one correspondence between a functional BSDE and a path-dependent PDE is established  for the Brownian case. 
This was then generalized  in~\citep{Ji2013} to the case of solutions to stochastic differential equations with functionally dependent drift and diffusion coefficients.  
We will now use~\citep[Theorem 20]{Ji2013} to formulate conditions such that  $(\text{FTVP})$ and $(\text{FTVP}^{+})$ admit unique solutions.
To this end, we will need the following regularity conditions from~\citep[Definition 3.1]{Peng2011}.
\begin{definition} 
  The functional $H:D([0,T],\bR^d) \mapsto \mathbb{R}$ on the Skorohod space $D([0,T],\bR^d)$ is of class $C^2(D([0,T],\bR^d))$ 
    if for all $\mb X\in D([0,T],\bR^d)$ and $t\in [0,T],$ there exist $\mb p_1\in \bR^d$ and $\mb p_2\in\bR^d\times\bR^d$ so that $\mb p_2$ is symmetric and the following holds
   $$
H(\mb X_{\mb X_t^h})-H(\mb X)=\mb p_1\cdot \mb h + \frac{1}{2}  \mb h^\top \mb p_2\mb h   + o(\arrowvert \mb h \arrowvert^2),\quad \mb h\in\bR^d, 
$$
where $\mb X_{\mb X_t^h}(u):=\mb X(u)I_{[0,t)}(u) + (\mb X(u) + \mb h)I_{[t,T]}(u). $ We denote $H'_{\mb X_t}(\mb X):=\mb p_1$ and $H''_{\mb X_t}(\mb X):=\mb p_2$.
Moreover, $H:D([0,T],\bR^d) \mapsto \mathbb{R}$ is of class $C_{l,lip}^2(D([0,T],\bR^d))$ if $H'_{\mb X_t}(\mb X)$ and $H''_{\mb X_t}(\mb X)$ exist for all $\mb X\in D([0,T],\bR_{(+)}^d)$ and $ t\in[0,T]$, and if there are constants $C,k>0$ 
  such that for all $\mb X,\mb Y\in D([0,T],\bR_{(+)}^d)$ (with $\Arrowvert\cdot\Arrowvert$ denoting the supremum norm),
  $$
  \arrowvert H(\mb X)-H(\mb Y)\arrowvert\le C(1+\Arrowvert\mb X\Arrowvert^k + \Arrowvert\mb Y\Arrowvert^k)\Arrowvert\mb X - \mb Y\Arrowvert,
    $$
  $$
   \arrowvert H'_{\mb X_t}(\mb X)-H'_{\mb Y_s}(\mb Y)\arrowvert\le C(1+\Arrowvert\mb X\Arrowvert^k + \Arrowvert\mb Y\Arrowvert^k)(\arrowvert t-s\arrowvert + \Arrowvert\mb X - \mb Y\Arrowvert),\; t,s\in [0,T]
  $$
   $$
   \arrowvert H''_{\mb X_t}(\mb X)-H''_{\mb Y_s}(\mb Y)\arrowvert\le C(1+\Arrowvert\mb X\Arrowvert^k + \Arrowvert\mb Y\Arrowvert^k)(\arrowvert t-s\arrowvert + \Arrowvert\mb X - \mb Y\Arrowvert),\; t,s\in [0,T].
$$
\end{definition}
  
  \begin{theorem}\label{JansonTysk thm F} Suppose that the terminal condition $H$ of $(\text{\rm FTVP}^{(+)})$ is of class $C^2_{l,lip}(D([0,T],\bR_{(+)}^d))$.    Then, under the following conditions, $(\text{\rm FTVP}^{(+)})$  admits a unique solution $F\in \bC^{1,2}([0,T))$.
      \begin{enumerate}
  \item  {\bf  (Theorem 20 in~\citep{Ji2013})} In case of $(\text{\rm FTVP})$, we suppose  that $a(t,\mb X(t))=\sigma(t,\mb X(t))\sigma(t,\mb X(t))^\top$ 
   with a Lipschitz continuous volatility matrix $\sigma.$ 
  \item In case of $(\text{\rm FTVP}^+)$,  we suppose  that $a(t,\mb X(t))=\sigma(t,\mb X(t))\sigma(t,\mb X(t))^\top$  with a Lipschitz continuous volatility matrix $\sigma$ such that 
  $a_{ii}(t,\mb X(t))$ is also Lipschitz continuous.
  \end{enumerate}
\end{theorem}

\begin{remark} 
Note that analogous conditions on the covariance, respectively, volatility structure, can also be formulated for the case where these quantities are path-dependent, thanks to~\citep[Theorem 20]{Ji2013}. However, for the purpose of this paper, which is establishing conditions on the covariance of the underlying under which no admissible arbitrage opportunities exist, we must stick to the choice of Markovian volatility 
in order to be able to apply a support theorem later on. 
\end{remark}
As above, the quantity $F_0(\mb S_0)$ can be identified as the amount of cash needed at $t=0$ so as to perfectly replicate the payoff $H$. 
But  in order to interpret $F_0(\mb S_0)$ as an arbitrage-free price, we have to know whether our class of trading strategies is indeed arbitrage-free. 
Below we will formulate Theorem~\ref{no arb thm functional}, which  is a  functional analogue of Theorem~\ref{no arb thm}. 

\begin{definition}\label{strategies def functional}
Suppose that the non-anticipative functional $F$ satisfying the conditions  from Definition~\ref{Def general adm integrands} is the solution of the  path-dependent heat equation  \eqref{heat eq functional} 
$$\mathcal{D} F+\cA^{(+)} F= 0\qquad\text{on } \bigcup_{t\in [0,T)}C([0,t], \bR^d_{(+)}).
$$ 
For $\mb S\in\cS_a^{(+)}$, we then define $\mb\xi^{\mb S}$ as in \eqref{xi ident eq functional} (on $[0,T)$) and 
\begin{equation}
V^{\mb S}_{\mb\xi}(t):=F_0(\mb S_0)+\int_0^t\mb\xi^{\mb S}(s)\ud \mb S(s).
\end{equation}
By  $\cY^{(+)}$ we denote the collection of all pairs $(F_0(\cdot),\mb \xi^\cdot)$ that arise in this way.
\end{definition}
The notion of an (admissible) arbitrage opportunity for $\cS_a^{(+)}$ in the functional setting is defined in analogy to Definition~\ref{arbitrage def}; we only have to replace $\cX^{(+)}$ by $\cY^{(+)}$.
\begin{theorem}[\bfseries Absence of admissible arbitrage]\label{no arb thm functional}
Suppose that $a(t,\mb X(t))$ is  continuous, bounded, and positive definite for all $(t,\mb X(t))\in[0,\wt T]\times\bR_{(+)}^d$, where $\wt T>T$. Then there are no admissible arbitrage opportunities in $\cY^{(+)}$. 
\end{theorem}

\section{Proofs}\label{Proofs section}

\subsection{Proofs of the results from Sections~\ref{pathwise hedging  section} and~\ref{arbitrage section}}
\begin{proof}[Proof of Proposition~\ref{most elementary Ito formula prop}] 
We first consider the case $i=j$. Then,
\begin{align*}
\mb\xi^{ii}(s)\cdot(\mb S(s')-\mb S(s))&=2(S_i(s)-K_{ii})(S_i(s')-S_i(s))\\
&=(S_i(s')-K_{ii})^2-(S_i(s)-K_{ii})^2-(S_i(s')-S_i(s))^2.
\end{align*}
Summing over $s\in\bT_n$ yields
\begin{equation}\label{quadr discrete Ito}
\sum_{s\in\bT_n,\,s\le  t}\mb\xi^{ii}(s)\cdot(\mb S(s')-\mb S(s))=( S_i(t_n)-K_{ii})^2-( S_i(0)-K_{ii})^2-\sum_{s\in\bT_n,\,s\le  t}(S_i(s')-S_i(s))^2,
\end{equation}
where $t_n=\max\{s'\,|\,s\in\bT_n,\, s\le t\}\searrow t$ as $n\ua\infty$. Clearly, the limit of the left-hand side exists if and only if the limit of the right-hand side exists, which implies the result for $i=j$. In case $i\neq j$,  the result follows just as above by using the already established existence of $\<S_k,S_k\>(t)$ for all $k$ and $t$ and by noting that $\sum_{k,\ell\in\{i,j\}}\<S_k,S_\ell\>=\<S_i+S_j,S_i+S_j\>$.\end{proof}

\begin{proof}[Proof of Proposition~\ref{sf prop}]
The pathwise It\^o formula yields that for $\mb S\in\cS_a^{(+)}$,
\begin{equation}\label{pathwise Ito in sf prop eq}
v(t,\mb S(t))=v(r,\mb S(r))+\int_r^t\nabla_{\mb x} v(s,\mb S(s))\ud \mb S(s)+\int_r^t \Big(\frac{\partial }{\partial t}v(s,\mb S(s))+\cL^{(+)}v(s,\mb S(s))\Big)\ud s.
\end{equation}
This immediately yields that (b) implies (a) and that \eqref{xi ident eq}
 must hold. 
 
 Let us now assume that (a) holds. Then 
 $$\int_r^t\big(\mb \xi^{\mb S}(s)-\nabla_{\mb x} v(s,\mb S(s))\big)\ud \mb S(s)=\int_r^t \Big(\frac{\partial }{\partial t}v(s,\mb S(s))+\cL^{(+)}v(s,\mb S(s))\Big)\ud s.
 $$
 Since the right-hand side has zero quadratic variation~\citep[Proposition 2.2.2]{Sondermann}, the same must be true of the left-hand side. By~\citep[Proposition 12]{SchiedCPPI}, the quadratic variation of the left-hand side is given by 
 $$\int_r^t\big(\mb \xi^{\mb S}(s)-\nabla_{\mb x} v(s,\mb S(s))\big)^\top a(s,\mb S(s))\big(\mb \xi^{\mb S}(s)-\nabla_{\mb x} v(s,\mb S(s))\big)\ud s
 $$
 in case of $\mb S\in\cS_a$.  Taking the derivative with respect to $t$ gives 
 $$\big(\mb \xi^{\mb S}(t)-\nabla_{\mb x} v(t,\mb S(t))\big)^\top a(t,\mb S(t))\big(\mb \xi^{\mb S}(t)-\nabla_{\mb x} v(t,\mb S(t))\big)=0$$
   for all $t$, and the fact that the matrix $a(t,\mb S(t))$ is positive definite yields that  \eqref{xi ident eq}
must hold. For $\mb S\in\cS_a^+$, the matrix  $ a(s,\mb S(s))$ needs to be replaced by the matrix with components $ a_{ij}(s,\mb S(s))S_i(s)S_j(s)$, and we arrive at  \eqref{xi ident eq} by the same  arguments as in the case of $\mb S\in\cS_a$. Plugging \eqref{xi ident eq}  into \eqref{pathwise Ito in sf prop eq} and using (a) implies that the rightmost integral in \eqref{pathwise Ito in sf prop eq} vanishes identically, which establishes (b) by again taking  the derivative with respect to $t$.
\end{proof}

  Now we prepare for the proof of Theorem~\ref{JansonTysk thm} (b).
   The following lemma can be proved by means of a straightforward computation. 
   
   \begin{lemma}\label{TVP transform lemma} For $\mb x=(x_1,\dots, x_d)^\top\in\bR^d$ let $\exp(\mb x):=(e^{x_1},\dots, e^{x_d})^\top\in\bR^d_+$. Then $v(t,\mb x)$ solves $(\text{TVP}^+)$ if and only if $\wt v(t,\mb x):=v(t,\exp(\mb x))$ solves 
$$
\begin{cases}\wt v\in C^{1,2}([0,T)\times \bR^d)\cap C([0,T]\times \bR^d),\\  \frac{\partial \wt v}{\partial t}+\wt\cL \wt v=0 \text{ in $[0,T)\times \bR^d$,}\\
\wt v(T,\mb x)=\wt f(\mb x) \text{ for $\mb x\in \bR^d$,}
\end{cases}\leqno{(\wt{\text{TVP}})}
$$
where $\wt f(\mb x)=f(\exp(\mb x))$ and
\begin{equation}\label{wt cL}
\wt\cL:=\frac12\sum_{i,j=1}^d\wt a_{ij}(t,\mb x)\frac{\partial^2}{\partial x_i\partial x_j}+\sum_{i=1}^d\wt b_i(t,\mb x)\frac{\partial}{\partial x_i},\qquad \mb x\in \bR^d,
\end{equation}
for $\wt a_{ij}(t,\mb x):=a_{ij}(t,\exp(\mb x))$ and $\wt b_i(t,\mb x):=-\frac12a_{ii}(t,\exp(\mb x))$.
   \end{lemma}

Next, the terminal-value problem  ${(\wt{\text{TVP}})}$ will be once again transformed into another auxiliary terminal-value problem. To this end, we need another transformation lemma, whose proof is also left to the reader.

    \begin{lemma}\label{wt TVP transform lemma} For $p>0$ let $g(\mb x):=1+\sum_{i=1}^de^{p x_i}$. Then $\wt v(t,\mb x)$ solves $(\wt{\text{TVP}})$ if and only if $\wh v(t,\mb x):=g(\mb x)^{-1}\wt v(t,\mb x)$ solves
     $$\begin{cases}\wh v\in C^{1,2}([0,T)\times \bR^d)\cap C([0,T]\times \bR^d),\\  \frac{\partial \wh v}{\partial t}+\wh\cL \wh v=0 \text{ in $[0,T)\times \bR^d$,}\\
\wh v(T,\mb x)=\wh f(\mb x) \text{ for $\mb x\in \bR^d$,}
\end{cases}\leqno{(\wh{\text{TVP}})}
$$
where $\wh f(\mb x)=\wt f(\mb x)/g(\mb x)$  and
\begin{equation}\label{wh cL}
\wh\cL:=\frac12\sum_{i,j=1}^d\wt a_{ij}(t,\mb x)\frac{\partial^2}{\partial x_i\partial x_j}+\sum_{i=1}^d\wh b_i(t,\mb x)\frac{\partial}{\partial x_i}+\wh c(t,\mb x),\qquad \mb x\in \bR^d,
\end{equation}
for 
\begin{align*}\wh b_i(t,\mb x)&=\wt b_i(t,\mb x)+pg(\mb x)^{-1}\sum_{j=1}^de^{px_j}\wt a_{ij}(t,\mb x),\\
\wh c(t,\mb x)&=\frac{p(p-1)}{2 g(\mb x)}\sum_{i=1}^d \wt a_{ii}(t,\mb x)e^{px_i}.
\end{align*}
      \end{lemma}

  \begin{proof}[Proof of Theorem~\ref{JansonTysk thm}]       We will show that ${(\wt{\text{TVP}})}$ admits a solution $\wt v$ if $|\wt f(\mb x)|\le c(1+\sum_{i=1}^de^{p x_i})$ for some $p>0$ and that $\wt v$ is unique in the class of functions that satisfy a similar estimate uniformly in $t$. 
   To this end, note that  the coefficients of $\wh \cL$ satisfy the conditions of~\citep[Theorem A.14]{JansonTysk}, i.e., $\wh a(t,\mb x)=\wt a(t,\mb x)$ is positive definite, there are constants $c_1$, $c_2$, $c_3$ such that for all $t$, $\mb x$, and $i,j$ we have that  $|\wt a_{ij}(t,\mb x)|\le c_1(1+|\mb x|^2)$, $|\wh b_i(t,\mb x)|\le c_2(1+|\mb x|)$, $|\wh c(t,\mb x)|\le c_3$, and $\wt a_{ij}$, $\wh b_i$, and $\wh c$ are locally H\"older continuous in $[0,T)\times\bR^d$. It therefore follows that ${(\wh{\text{TVP}})}$ admits a unique bounded solution $\wh v$ whenever $\wh f$ is bounded and continuous. But then $\wt v(t,\mb x):=g(\mb x)\wh v(t,\mb x)$ solves ${(\wt{\text{TVP}})}$ with terminal condition $\wt f(\mb x):=g(\mb x)\wh f(\mb x)$. Hence, ${(\wt{\text{TVP}})}$ admits a solution whenever $|\wt f(\mb x)|\le c(1+\sum_{i=1}^de^{p x_i})$ for some $p>0$.  Lemma~\ref{TVP transform lemma} now establishes the existence of solutions to ${({\text{TVP}^+})}$
if the terminal condition is continuous and has at most polynomial growth.
\end{proof}

\begin{remark}\label{growth remark}It follows from the preceding argument that, if $|f(\mb x)|\le c(1+|\mb x|^p)$, then the corresponding solution $v$ of ${({\text{TVP}^+})}$ satisfies $|v(t,\mb x)|\le \wt c(1+|\mb x|^p)$ for a certain constant $\wt c$ and with the same exponent $p$.
\end{remark}

\begin{proof}[Proof of Theorem~\ref{exotic Janson Tysk thm}] 
We first prove the result in the case of $(\text{TVP})$. The function $v_k$ will be well-defined if $f_{k+1}$ is continuous and has at most polynomial growth. It is easy to see that these two properties will follow if $v_{k+1}$ satisfies the following three conditions:
\begin{enumerate}[(i)]
\item $(\mb x_0,\dots, \mb x_{k+1},\mb x)\mapsto v_{k+1}(t,\mb x_0,\dots, \mb x_{k+1},\mb x)$ has at most polynomial growth;
\item
$\mb x\mapsto v_{k+1}(t_{k+1},\mb x_0,\dots, \mb x_{k+1},\mb x)$ is continuous for all $\mb x_0,\dots, \mb x_{k+1}$;
\item
$(\mb x_0,\dots, \mb x_{k+1})\mapsto v_{k+1}(t,\mb x_0,\dots, \mb x_{k+1},\mb x)$ is locally Lipschitz continuous, 
uniformly in $t$ and locally uniformly in $\mb x$, with a Lipschitz constant that grows at most polynomially. More precisely, there exist $p\ge0$ and $L\ge0$ such that, for  $|\mb x|,|\mb x_i|,|\mb y_i|\le m$ and  $t\in[t_{k+1},t_{k+2}]$, 
$$\Big| v_{k+1}(t,\mb x_0,\dots,  \mb x_{k+1},\mb x)-v_{k+1}(t,\mb y_0,\dots, \mb y_{k+1},\mb x)\Big|\le (1+m^p)L\sum_{i=0}^{k+1}|\mb x_i-\mb y_i|.
$$
\end{enumerate}

We will now show that $v_k$ inherits properties (i), (ii), and (iii) from $v_{k+1}$. Since these properties are obviously satisfied by $v_N$, the assertion will then follow  by backward induction.

To establish (i), let  $p,c>0$ be such that $\wt f_{k+1}(\mb x):=c(|\mb x_0|^p+\cdots+|\mb x_k|^p+|\mb x|^p+|\mb x|^p)$ satisfies
$-\wt f_{k+1}\le f_{k+1}\le \wt f_{k+1}$. Then let $\wt v_k(t,\mb x_0,\dots, \mb x_{k},\mb x)$ be the solution of  $(\text{TVP})$ with terminal condition $\wt f_{k+1}$ at time $t_{k+1}$.  Theorem~\ref{JansonTysk thm},~\citep[Theorem A.7]{JansonTysk}, and the linearity of solutions imply that $(\mb x_0,\dots, \mb x_{k},\mb x)\mapsto \wt v_{k}(t,\mb x_0,\dots, \mb x_{k},\mb x)$ has at most polynomial  growth, while the maximum principle in the form of~\citep[Theorem A.5]{JansonTysk} implies that $-\wt v_k\le v_k\le \wt v_k$. This establishes (i).

Condition (ii) is satisfied automatically, as solutions to $(\text{TVP})$ are continuous by construction.

To obtain (iii), let $p$ and $L$ be as in (iii) and $\mb x_i,\mb y_i$ be given. We take $m$ so that $m\ge |\mb x_i|\vee|\mb y_i|$
for $i=1,\dots, k$ and let $\delta:=L\sum_{i=0}^k|\mb x_i-\mb y_i|$. Then
\begin{align*}
-(1+m^p+|\mb x|^p)\delta\le v_{k+1}(t_{k+1},\mb x_0,\dots, \mb x_k,\mb x,\mb x)-v_{k+1}(t_{k+1},\mb y_0,\dots, \mb y_k,\mb x,\mb x) \le (1+m^p+|\mb x|^p)\delta.
\end{align*}
Now we define $u(t,\mb x)$ as the solution of (TVP) with terminal condition $u(t_{k+1},\mb x)=|\mb x|^p$ at time $t_{k+1}$. Theorem~\ref{JansonTysk thm} implies that $u$ is well defined, and the maximum principle and~\citep[Theorem A.7]{JansonTysk} imply that $0\le u(t,\mb x)\le c|\mb x|^p$ for some constant $c\ge0$. Another application of the maximum principle yields that 
$$-(1+m^p+u(t,\mb x))\delta\le v_{k}(t,\mb x_0,\dots, \mb x_k,\mb x)-v_{k}(t,\mb y_0,\dots, \mb y_k,\mb x)\le (1+m^p+u(t,\mb x))\delta
$$
 for all $t$ and $\mb x$, which establishes that (iii) holds for $v_k$ with  the same $p$ and the new Lipschitz constant $(1+c)L$.
  
Now we turn to the proof in case of $(\text{TVP}^+)$.  It is clear from our proof of Theorem~\ref{JansonTysk thm} (b) that $(\text{TVP}^+)$ inherits the maximum principle from $(\wh{\text{TVP}})$. Moreover, Remark~\ref{growth remark}  shows that $v_k$ inherits property (i) from $v_{k+1}$. So Remark~\ref{growth remark} can replace 
\citep[Theorem A.7]{JansonTysk} in the preceding argument. Therefore, the proof for $(\text{TVP}^+)$ can be carried out in the same way as for  $(\text{TVP})$.
\end{proof}

\begin{proof}[Proof of Theorem~\ref{no arb thm}] We first prove the result in case of $\cX$.  Let us suppose by way of contradiction that there exists an admissible arbitrage opportunity in $\cX$, and let $0=t_0<t_1<\cdots <t_N=t_{N+1}=T$ and $v_k$ denote the corresponding time points and functions as in Definition~\ref{strategies def}.

Under our assumptions, the martingale problem for the operator $\cL$ is well-posed~\citep{StroockVaradhanMartingaleProblem}. Let $\bP_{t,\mb x}$ denote the corresponding Borel probability measures on  $C([t,T],\bR^d)$ under which the coordinate process, $\left(\mb X(u)\right)_{t\leq u\leq T}$,  is a diffusion process with generator $\cL$ and satisfies $X(t)=x$ $\bP_{t,\mb x}$-a.s. In particular, $X_i$ is a continuous local $\bP_{t,\mb x}$-martingale for $i=1,\dots, d$.
Moreover, the support theorem~\citep[Theorem 3.1]{StroockVaradhanSupport} states that  the law of $(\mb X(u))_{t\le u\le T}$ under $\bP_{t,\mb x}$ has full support on $C_{\mb x}([t,T],\bR^d):=\{\mb\om \in C([t,T],\bR^d)\,|\,\mb\om(t)=\mb x\}$.

In a first step, we now use these facts to show that all functions $v_k$ are nonnegative. To this end, we note first that the support theorem implies that the law of $(\mb X(t_1),\dots,\mb X(t_N))$ under $\bP_{0,\mb x}$ has full support on $(\bR^d)^N$. Since $\bP_{0,\mb x}$-a.e.~trajectory in  $C_{\mb x}([0,T],\bR^d)$ belongs to $\cS_a$, it follows that the set $\big\{(\mb S(t_1),\dots,\mb S(t_N))\,|\,\mb S\in\cS_a,\ \mb S(0)=\mb x\big\}$ is dense in $(\bR^d)^N$. Condition~(a) of Definition~\ref{arbitrage def} and the continuity of $v_N$ thus imply that $v_N(T,\mb x_0,\dots, \mb x_{N+1})\ge0$ for all $\mb x_0,\dots, \mb x_{N+1}$. In the same way, we get from the admissibility of the arbitrage opportunity that  $v_k(t,\mb x_0,\dots,\mb x_k,\mb x)\ge -c$ for all $k$, $t\in[t_k,t_{k+1}]$ and $\mb x_0,\dots,\mb x_k,\mb x\in\bR^d$. 

For the moment, we  fix $\mb x_0,\dots,\mb x_{N-1}$ and consider the function $u(t,\mb x):=v_{N-1}(t,\mb x_0,\dots,\mb x_{N-1},\mb x)$. Let $Q\subset\bR^d$ be a bounded domain whose closure is contained in $\bR^d$ and let $\tau:=\inf\{s\,|\,\mb X(s)\notin Q\}$ be the first  exit time from $Q$. By It\^o's formula and the fact that $u$ solves $(\text{TVP})$ we have $\bP_{t,\mb x}$-a.s.~for $t\in[t_{N-1},T)$ that
\begin{equation}
u(T\wedge\tau,\mb X({T\wedge\tau}))=u(t,\mb x)+\int_t^{T\wedge\tau}\nabla_{\mb x}u(s,\mb X(s))\ud \mb X(s).
\end{equation}
Since $\nabla_{\mb x}u$ and the coefficients of $\cL$ are bounded in the closure of $Q$, the stochastic integral on the right-hand side is a true martingale. Therefore,
\begin{equation}\label{u(t,x) rep eq}
u(t,\mb x)=\bE_{t,\mb x}[\,u(T\wedge\tau,\mb X({T\wedge\tau}))\,].
\end{equation}
Now let us take an increasing sequence $Q_1\subset Q_2\subset\cdots$ of bounded domains exhausting $\bR^d$ and whose closures are  contained in $\bR^d$. By $\tau_n$ we denote the exit time from $Q_n$. Then, an application of \eqref{u(t,x) rep eq} for each $\tau_n$, Fatou's lemma in conjunction with the fact that $u\ge-c$, and the already established nonnegativity of $u(T,\cdot)$  yield 
\begin{equation}\label{localization eq}
u(t,\mb x)=\lim_{n\uparrow\infty}\bE_{t,\mb x}[\,u(T\wedge\tau_n,\mb X({T\wedge\tau_n}))\,]\ge \bE_{t,\mb x}[\,u(T,\mb X({T}))\,]\ge0.
\end{equation}
This establishes the nonnegativity of $v_{N-1}$ and in particular of the terminal condition $f_{N-1}$ for $v_{N-2}$. We may therefore  repeat the preceding argument for $v_{N-2}$  and so forth.  Hence,  $v_k\ge0$ for all~$k$. 

Now let $\mb S\in\cS_a$ and $T_0$ be such that $V^{\mb S}_{\mb\xi}(0)\le0$ and $V^{\mb S}_{\mb\xi}(T_0)>0$, which exists according to the assumption made at the beginning of this proof. If $k$ is such that $t_k< T_0\le t_{k+1}$ and $\mb x_0:=\mb S(0)$, then $v_0(0,\mb x_0)=0$ and $v_k(T_0,\mb S(t_0),\dots,\mb S(t_k),\mb S(T_0))>0$. By continuity, we actually have $v_k(T_0,\cdot)>0$ in  an open neighborhood $U\subset C_{\mb x}([0,T],\bR^d)$ of the path $\mb S$.

 Since $\bP_{0,\mb x_0}$-a.e.~sample path belongs to $\cS_a$, It\^o's formula gives that $\bP_{0,\mb x_0}$-a.s.,
$$v_k(T_0,\mb X(t_0),\dots,\mb X(t_k),\mb X(T_0))=v_0(0,\mb x_0)+\int_0^{T_0}\mb\xi^{\mb X}(t)\ud\mb X(t).
$$
Localization as in \eqref{localization eq} and using the fact that $v_\ell\ge0$ for all $\ell$ implies that 
$$0=v_0(0,\mb x_0)\ge\bE_{0,\mb x_0}\big[\,v_k(T_0,\mb X(t_0),\dots,\mb X(t_k),\mb X(T_0))\,\big]\ge0.
$$
Applying once again the support theorem now yields a contradiction to the fact that $v_k(T_0,\cdot)>0$ in  the open set $U$. This completes the proof for $\cX$. 

Now we turn to the proof for $\cX^+$. In this case, the martingale problem for the operator $\wt \cL$ defined in \eqref{wt cL} is well posed since the coefficients of $\wt \cL$ are again bounded and continuous~\citep{StroockVaradhanMartingaleProblem}. These properties of the coefficients also guarantee that  the support theorem holds~\citep[Theorem 3.1]{StroockVaradhanSupport}. If $(\wt \bP_{s,\mb x},\wt{\mb X})$ is a corresponding diffusion process, we can consider the laws of $\mb X(t):=\exp(\wt{\mb{X}}(t))$ and, by Lemma~\ref{TVP transform lemma}, obtain a solution to the martingale problem for $\cL^+$, which satisfies the support theorem with state space $\bR^d_+$. We can now simply repeat the arguments from the proof for $\cX$ to also get the result for $\cX^+$.
 \end{proof}

 \subsection{Proofs of the results from Section~\ref{Functional Section}}
 \begin{proof}[Proof of Proposition~\ref{sf prop functional}]
 The proof is analogous to the  proof of Proposition~\ref{sf prop}. 
For  $ \mb S\in\cS_a$, all that is needed in addition to the arguments of Proposition~\ref{sf prop} is the fact that the quadratic variation of 
$$\int_r^t\big(\mb \xi^{\mb S}(s)-\nabla_{ x} F_s(\mb S_{[r,u],s}) \big)\ud  \mb S(s)$$
is given by
 $$\int_r^t\big( \mb \xi^{ \mb S}(s)-\nabla_{ x} F_s(\mb S_{[r,u],s})\big) ^\top  a(s, \mb S(s))\big( \xi^{ \mb S}(s)-\nabla_{ x} F_s(\mb S_{[r,u],s})\big) \ud s;
 $$
see~\citep[Proposition 2.1]{SchiedVoloshchenkoAssociativity}.  For $ \mb S\in\cS_a^+$, the matrix  $ a(s, \mb S(s))$ has to be replaced by the matrix with components $ a_{ij}(s, \mb S(s))S_i(s)S_j(s)$. \end{proof}

 To prove Theorem~\ref{JansonTysk thm F} and Theorem~\ref{no arb thm functional} we need the following lemma, which is a straightforward extension of Lemma~\ref{TVP transform lemma} 
to the functional setting. Its proof is therefore left to the reader. For $\mb X$ in the Skorohod space $D([0,T],\bR^d)$ 
 we set $\left(\exp(\mb X)\right)_t
 =\exp(\mb X_t):=\left(\exp(\mb X(u))\right)_{0\le u\le t}\in D([0,t],\bR^d_+)$.

 \begin{lemma}\label{TVP transform lemma functional}  The functional $F_t(\mb X_t)$  solves $( \text{FTVP}^{+})$
    if and only if $\wt F_t(\mb X_t):=F_t(\exp(\mb X_t))$ solves 
$$
\begin{cases}\wt F\in\mathbb{C}^{1,2}([0,T))&
\text{satisfies the conditions  from Definition~\ref{Def general adm integrands},}\\ \mathcal{D} \wt F+\wt \cA \wt F=0&\text{in $\bigcup_{t\in[0,T)}D([0,t], \bR^d),$}\\
\wt F_T(\mb X_T)=\wt H(\mb X_T) &\text{for $\mb X_T\in D([0,T],\bR^d)$,}
\end{cases}\leqno{(\wt{\text{FTVP}})}
$$
where $\wt H(\mb X_T)=H(\exp(\mb X_T))$ 
 and
\begin{equation}\label{wt cL functional}
\wt\cA:=\frac12\sum_{i,j=1}^d\wt a_{ij}(t,\mb X(t))\nabla^2_{ij}+\sum_{i=1}^d\wt b_{i}(t,\mb X(t))\partial_{i}\qquad \text{in $\bigcup_{t\in[0,T)}D([0,t], \bR^d),$}
\end{equation}
where, as in~\cite[Eq.~(15)]{CF}, $\partial_i$ are the partial   vertical derivatives,
 $\wt a_{ij}(t,\mb X(t)):=a_{ij}(t,\exp(\mb X(t)))$, and $\wt b_{i}(t,\mb X(t)):=-\frac12a_{ii}(t,\exp(\mb X(t)))$.
   \end{lemma}
\ignore{
     \begin{lemma}\label{wt TVP transform lemma functional} For $p>0$ let $G(\mb X(t)):=1+\sum_{i=1}^de^{p X_{i}(t)}$. Then $\wt F_t(\mb X_t)$ solves $(\wt{\text{FTVP}})$ if and only if $\wh F_t(\mb X_t):=G(\mb X(t))^{-1}\wt F_t(\mb X_t)$ solves
     $$\begin{cases}\wh F\in \mathbb{C}^{1,2}([0,T))&\text{satisfies the conditions  from Definition~\ref{Def general adm integrands}},\\  
 \mathcal{D} \wh F+\wh \cA \wh F=0&\text{in $\bigcup_{t\in[0,T)}D([0,t], \bR^d),$}\\
\wh F_T(\mb X_T)=\wh H(\mb X_T)& \text{for $\mb X_T\in D([0,T],\bR^d)$,}
\end{cases}\leqno{(\wh{\text{FTVP}})}
$$
where $\wh H(\mb X_T)=\wt H(\mb X_T)/G(\mb X(T))$  and
\begin{equation}\label{wh cA}
\wh\cA:=\frac12\sum_{i,j=1}^d\wt a_{ij}(t,\mb X(t))\nabla^2_{ij}+\sum_{i=1}^d\wh b_i(t,\mb X(t))\partial_{i}+\wh c(t,\mb X(t))\qquad \text{in $\bigcup_{t\in[0,T)}D([0,t], \bR^d),$}
\end{equation}
for 
\begin{align*}\wh b_i(t,\mb X(t))&=\wt b_i(t,\mb X(t))+pG(\mb X(t))^{-1}\sum_{j=1}^de^{pX_j(t)}\wt a_{ij}(t,\mb X(t)),\\
\wh c(t,\mb X(t))&=\frac{p(p-1)}{2 G(\mb X(t))}\sum_{i=1}^d \wt a_{ii}(t,\mb X(t))e^{pX_i(t)}.
\end{align*}
      \end{lemma}}
      Note that 
      the chain rule for functional derivatives (see~\citep[p.6]{Dupire}) 
      implies the equivalence of the PDEs in $(\wt{\text{FTVP}})$ and $(\text{FTVP}).$ 
      \ignore{Analogously, the PDEs in $(\wt{\text{FTVP}})$ and $(\wh{\text{FTVP}})$ are equivalent, due to the product rule for functional derivatives.}
       Regarding the regularity conditions in Definition~\ref{Def general adm integrands}, we note that $\wt F$ will be regular enough if and only if $F$ is regular enough (because $\exp(\mb X(t))$ is a sufficiently regular functional). 
       \ignore{Similarly, $\wh F$ will be regular in the sense of Definition~\ref{Def general adm integrands} if and only if $\wt F$ is.} 
     
  \begin{proof}[Proof of Theorem~\ref{JansonTysk thm F}]       
  Part (a) directly follows from~\citep[Theorem 20]{Ji2013}. 
  To prove part (b), note that  the coefficients of $\wt \cA$ satisfy the conditions of~\citep[Theorem 20]{Ji2013}, i.e., 
    $\wt a(t,\mb X(t))$ is positive definite 
  and can be written as $\wt \sigma(t,\mb X(t))\wt \sigma(t,\mb X(t))^\top$ with a Lipschitz continuous volatility coefficient $\wt \sigma,$ 
     and $\wt b_i$ is also  Lipschitz. 
     It therefore follows that ${(\wt{\text{FTVP}})}$ admits a unique solution $\wt F\in \bC^{1,2}([0,T))$ if $\wt H\in C^2_{l,lip}(D([0,T],\bR^d)).$
  \ignore{But then $\wt F(t,\mb X_t):=G(\mb X(t))\wh F(t,\mb X_t)$ solves ${(\wt{\text{FTVP}})}$ with terminal condition $\wt H(\mb X_T):=G(\mb X(T))\wh H(\mb X_T).$ 
  Hence, ${(\wt{\text{FTVP}})}$ admits a unique solution whenever $\wt H\in C^2_{l,lip}(D([0,T],\bR^d))$.}  
   Lemma~\ref{TVP transform lemma functional} now establishes the existence of solutions to ${({\text{FTVP}^+})}$ if the terminal condition is of class $C^2_{l,lip}(D([0,T],\bR_{+}^d))$.
   \end{proof}

\begin{proof}[Proof of Theorem~\ref{no arb thm functional}] 
The proof is similar to the one of Theorem~\ref{no arb thm}. We first  consider the case of $\cY$. Let $\mb X$ and $\bP_{t,\mb x}$ ($0\le t\le T$, $\mb x\in\bR^d$) be 
as in the proof of Theorem~\ref{no arb thm}. For a path $\mb Y\in C([0,T],\bR^d)$, we define $\bbar\bP_{t,\mb Y_t}$ as that probability measure on  $C([0,T],\bR^d)$ under which the coordinate process $\mb X$ satisfies $\bbar\bP_{t,\mb Y_t}$-a.s.~$\mb X(s)=\mb Y(s)$ for $0\le s\le t$ and under which the law of $(\mb X(u))_{t\le u\le T}$ is equal to $\bP_{t,\mb Y(t)}$. The  support theorem~\cite[Theorem 3.1]{StroockVaradhanSupport} then states that  the law of $(\mb X(u))_{0\le u\le T}$ under $\bbar\bP_{t,\mb Y_t}$ has full support on $C_{\mb Y_t}([0,T],\bR^d):=\{\mb\om \in C([0,T],\bR^d)\,|\,\mb\om_t=\mb Y_t\}$.

Now suppose by way of contradiction that there exists an admissible arbitrage opportunity arising from a non-anticipative functional  $F$ as in Definition~\ref{strategies def functional}. 
In a first step, we  show that $F$ is nonnegative on $[0,T]\times C([0,T],\bR^d)$. As in the  proof of Theorem~\ref{no arb thm}, the support theorem implies that  $\big\{(\mb S(t))_{0\leq t\leq T}\,|\,\mb S\in\cS_a,\ \mb S(0)=\mb x\big\}$ is dense in  $C_{\mb x}([0,T],\bR^d)$.  Condition~(a) of Definition~\ref{arbitrage def} and the left-continuity of $F$ in the sense of~\cite[Definition 3]{CF}
 thus imply that $F_T(\mb Y)\ge0$ for all $\mb Y\in C([0,T],\bR^d)$. In the same way, we get from the admissibility of the arbitrage opportunity that  $F_t(\mb Y_t)\ge -c$ 
  for all $t\in[0,T]$ and $\mb Y\in C([0,T],\bR^d)$. 
To show that actually $F_t(\mb Y_t)\ge 0$, let $Q\subset\bR^d$ be a bounded domain whose closure is contained in $\bR^d$ and let $\tau:=\inf\{s\,|\,\mb X(s)\notin Q\}$ be the first  exit time from $Q$. By the functional change of variables formula, in conjunction with  the fact that $F$ solves $(\text{FTVP})$ (on continuous paths), we obtain $\bbar\bP_{t,\mb Y_t}$-a.s.~for $t\in[0,T)$ that
\begin{equation}\label{stochint}
F_{T\wedge\tau}(\mb X_{T\wedge\tau})=F_t(\mb Y_t)+\int_t^{T\wedge\tau}\nabla_{\mb x}F_s(\mb X_s)\ud \mb X(s).
\end{equation}
By~\citep[Proposition 2.1]{SchiedVoloshchenkoAssociativity}, 
we have
$$
\Big \langle\int_t^{\cdot\wedge\tau}\nabla_{\mb x}F_s(\mb X_s)\ud \mb X(s)\Big\rangle(T)= \int_t^{T\wedge\tau}\nabla_{\mb x}F_s(\mb X_s)^\top  a(s, \mb X(s))\nabla_{\mb x}F_s(\mb X_s) \ud s.
 $$
 Since $\nabla_{\mb x}F$ and the coefficients of $\cA$ are bounded in the closure of $Q$, the stochastic integral on the right-hand side of \eqref{stochint} is a true martingale. 
Therefore,
\begin{equation}\label{u(t,x) rep eq functional}
F_t(\mb Y_t)=\bbar\bE_{t,\mb Y_t}[\,F_{T\wedge\tau}(\mb X_{T\wedge\tau})\,].
\end{equation}
Now let us take an increasing sequence $Q_1\subset Q_2\subset\cdots$ of bounded domains exhausting $\bR^d$ and whose closures are contained in $\bR^d$. By $\tau_n$ we denote the exit time from $Q_n$. Then, an application of \eqref{u(t,x) rep eq functional} for each $\tau_n$, Fatou's lemma in conjunction with the fact that $F\ge-c$, and the already established nonnegativity of $F_T(\cdot)$ 
 yield 
\begin{equation}\label{localization eq functional}
F_t(\mb Y_t)=\lim_{n\uparrow\infty}\bbar\bE_{t,\mb Y_t}[\,F_{T\wedge\tau_n}(\mb X_{T\wedge\tau_n})\,]\ge \bbar\bE_{t,\mb Y_t}[\,F_T(\mb X_{T})\,]\ge0.
\end{equation}
This establishes the nonnegativity of $F$ on $[0,T]\times C([0,T],\bR^d)$.

Now let $\mb S\in\cS_a$ and $T_0$ be such that $V^{\mb S}_{\mb\xi}(0)\le0$ and $V^{\mb S}_{\mb\xi}(T_0)>0$. Since $V^{\mb S}_{\mb\xi}(t)=F_t(\mb S_t)$ by Proposition~\ref{sf prop functional}, we have $F_0(\mb S_0)=0$ and $F_{T_0}(\mb S_{T_0})>0$. By left-continuity of $F$, 
we actually have $F_{T_0}(\cdot)>0$ 
 in  an open neighborhood $U\subset C_{\mb S(0)}([0,T],\bR^d)$ of the path $\mb S$.

Since $\bP_{0,\mb S(0)}$-a.e.~sample path belongs to $\cS_a$, the functional change of variables formula gives that $\bP_{0,\mb S(0)}$-a.s.,
$$F_{T_0}(\mb X_{T_0})=F_0(\mb S_0)+\int_0^{T_0}\mb\xi^{\mb X}(t)\ud\mb X(t).
$$
Localization as in \eqref{localization eq functional} and using the fact that $F\ge0$ implies that 
$$0=F_0(\mb S_0)\ge\bE_{0,\mb S(0)}\big[\,F_{T_0}(\mb X_{T_0})\,\big]\ge0.
$$
Applying once again the support theorem now yields a contradiction to the fact that $F_{T_0}(\cdot)>0$ in  the open set $U$. This completes the proof for $\cY$. 

The proof for $\cY^+$ is completed by an exponential transformation, as in the proof of Theorem~\ref{no arb thm}. 
\end{proof}

 \setlength{\bibsep}{1.8pt}

 \bibliography{CTBook}{}
\bibliographystyle{abbrv}

\end{document}